\documentclass[12pt]{article}
\usepackage{graphicx}% Include figure filse
\usepackage{dcolumn}% Align table columns on decimal point
\usepackage{amsmath,amsthm,amssymb}

\newtheorem{conj}{Conjecture}
\newtheorem{thm}{Theorem}
\newtheorem{prop}{Proposition}
\newtheorem{defn}{Definition}
\newtheorem{cor}{Corollary}
\newtheorem{lem}{Lemma}

\def\wt{\widetilde}
\def\wh{\widehat}

\title{Initial data for rotating cosmologies}
\author{Piotr Bizo\'n$^{1}$, Stefan Pletka$^{2}$ and Walter Simon$^{2}$}
 \date{${}^{1}$ Institute of Physics, Jagiellonian University, Cracow, Poland\\${}^{2}$Gravitational Physics, Faculty of
 Physics, University of Vienna, Austria}

\begin{document}
\sloppy
\maketitle
\begin{abstract}

We revisit the construction of maximal initial data on compact manifolds in vacuum with positive cosmological
constant via the conformal method. We discuss, extend and apply recent results of Hebey et al. 
\cite{HPP} and Premoselli \cite{BP} which yield existence, non-existence, (non-)uniqueness and (linearisation-) 
stability of solutions of the Lichnerowicz equation, depending on
its coefficients. 
We then focus on so-called $(t,\varphi)$-symmetric data as ``seed manifolds'', and in particular on 
Bowen-York data on the round hypertorus $\mathbb{S}^2 \times \mathbb{S}$ (a slice of Nariai)  and
on Kerr-de Sitter.
In the former case, we clarify the bifurcation structure of the axially
symmetric solutions of the Lichnerowicz equation in terms of the angular 
momentum as bifurcation parameter, using a combination of analytical and numerical techniques.
As to the latter example, we show how dynamical data can be constructed
in a natural way via conformal rescalings of  Kerr-de Sitter data.

\end{abstract}

\section{Introduction}

We start with two definitions.

\begin{defn}
\label{ids}
As {\bf Initial Data (ID)} $(\wt {\cal M}, \wt g_{ij}, \wt K_{ij})$
(i,j,=1,2,3) for vacuum with positive cosmological constant $\Lambda$ we take a compact 3-dim. Riemannian manifold 
$\widetilde {\cal M}$ with smooth metric $\widetilde g_{ij}$ and smooth second fundamental form $\widetilde K_{ij}$
which is maximal  $\widetilde g^{ij} \widetilde K_{ij} = 0$ and satisfies the constraints 
\begin{equation}
\label{con}
\widetilde R   = \widetilde K_{ij} \widetilde K^{ij} + 2\Lambda \qquad
\widetilde \nabla_i \widetilde K^{ij} = 0
\end{equation}
Here $\widetilde \nabla$ and $\widetilde R$ are the covariant derivative and
the scalar curvature of $\widetilde g_{ij}$.
\end{defn}
\begin{defn}
\label{seed}
A {\bf Seed Manifold (SM)} $({\cal M}, g_{ij}, K_{ij})$ consists of a compact
3-dim. manifold $({\cal M},g_{ij})$ with smooth metric in the positive Yamabe class, 
and of a smooth trace-free and divergence-free tensor $K_{ij}$ on ${\cal M}$.
\end{defn} 

The  conformal method is the art of turning a SM into an ID via conformal  rescaling
\cite{JIM}.
 In the present setting it remains to be shown that the Lichnerowicz equation

\begin{equation}
\label{phi}
 \Phi(\phi):=  - \left( \Delta - \frac{1}{8} R \right) \phi - \frac{1}{4}\Lambda \phi^5 - \frac{\Omega^2}{8 \phi^{7}} = 0  
\end{equation}
has a smooth, strictly positive solution $\phi$, where $\Delta$ and $R$ are the
Laplacian and the scalar curvature of $g_{ij}$, and $\Omega^2 = K^{ij}K_{ij}$.
 In this case the  ``physical'' quantities 
\begin{equation}
\label{ct}
\widetilde g_{ij} = \phi^4 g_{ij} \qquad \widetilde K^{ij} = \phi^{-10} K^{ij}
\end{equation}
 indeed satisfy the constraints (\ref{con}).

In view of the observed small positive value of $\Lambda$, and due to the naturality of the assumption of maximality of the data, 
we are dealing here with a physically very realistic case of the Lichnerowicz equation.  
It is precisely this case, however, which involves rather intricate mathematical
 problems.  Firstly, solutions definitely do {\it not} exist for {\it large} 
$\Omega^2$ which is rather easy to see in principle either from the maximum principle, or by integrating (\ref{phi}). 
On the other hand, existence proofs for {\it small} $\Omega^2$ are subtle, in
particular when  $\Omega^2$ is allowed to have zeros \cite{EH,HPP,HV,BP}. 
However, in physically meaningful situations $\Omega^2$ does have zeros - in the axially symmetric
(AS)  case, on which we focus in this paper and which is simple in other respects, 
$\Omega^2$ in fact typically vanishes on the axis (cf. Sect. 3).

There are now available two types of general existence and non-existence results which cover 
the case of present interest. The first one, due to Hebey et al. \cite{HPP} (see also \cite{EH,HV}) guarantees existence 
of solutions if $\int_{\cal M}\Omega^2$ is small, and proves
non-existence if $\int_{\cal M} \Omega^{5/6}$ is large. In either case, the bounds can be given
explicitly in terms of the Yamabe constant of ${\cal M}$ and other integrals over
${\cal M}$. However, there is a $\Omega$-``gap''  which is not covered by these
results. In the second theorem, due to Premoselli \cite{BP}, $\Omega$ is written
as $\Omega = b \, \Omega_0$ for some (fixed) function $\Omega_0$ and (variable) 
constant $b > 0$, and the result is ``gap free'': It is asserted that there is a 
constant $b_{\ast} \in (0,\infty)$ such that (\ref{phi}) has at least two positive solutions for all 
$b < b_{\ast}$, a unique solution for $b = b_{\ast}$ and no solution for $b >
b_{\ast}$.
Moreover,  for  every $b \le b_{\ast}$ there is a unique stable, ``minimal'' solution. 
We remark, however, that in this theorem there is no 
direct information about $b_{\ast}$ in terms of more familiar geometric quantities of ${\cal M}$. 

In this work we start (in Sect. 2.1) with defining (in Def. 3) (linearisation-) 
stability of solutions of (\ref{phi}) and of ID (under conformal deformations), 
which will be key in what follows. 
In particular we prove Proposition 1 which guarantees instability if
$\wt \Omega^2 = \wt K_{ij} \wt K^{ij} < \Lambda$. Another important issue in our work is 
``symmetry-inheriting" versus "symmetry-breaking" of solutions, by which we mean solutions of (\ref{phi})
which share (or don't share) all symmetries of the equation. In Sect. 2.2. we prove 
a simple result (Proposition 2) which ensures symmetry inheritance for stable solutions in the
case of continuous symmetries. We proceed  in Sect. 2.3. by reviewing the theorems of Hebey et al. and
 Premoselli mentioned above. As  small complements to the latter result, 
we clarify (in Proposition 3) how the stable, minimal solutions of
 (\ref{phi}) approach zero
 as $b \rightarrow 0$. Moreover, in Proposition 4 we   
employ an argument from  bifurcation theory to show that near the 
maximal value $b_*$, there are precisely two solutions.
    
The core of our paper is Sect. 3  where we apply the results sketched above
to certain ``$(t,\varphi)$-symmetric data'' as introduced and discussed in
\cite{SD,DLT}. 
There one sets out from an AS ``twist potential'' $\omega$ from
which there is constructed an AS, symmetric, trace-free and divergence free tensor
$K_{ij}$.  The SM constructed in this way ``rotate'' in general, and the (Komar-) angular momentum $J$ of any
selected 2-surface is given directly in terms of the values of $\omega$ on the axis
via  $8\pi J = [\omega(0) - \omega(\pi)]$.  
We focus on two different classes of seed manifolds (${\cal M}, g_{ij}$) as examples:  
In Sect. 3.3. we consider a ``round hypertorus'',  i.e. $\mathbb{S}^2 \times \mathbb{S}^1$ 
with a round $\mathbb{S}^2$. We first review the case without angular momentum 
where the solutions of (\ref{phi}) yield the time symmetric Kottler
(Schwarzschild-de Sitter) data. Then we consider a $K_{ij}$ of ``Bowen-York form''
\cite{RB,BY} as the simplest non-trivial rotating model. 
Applying the results of Hebey et al. \cite{HPP} we find (in Theorem 3) that small angular momenta (compared to
$\Lambda^{-1}$, and taken w.r.t. the $\mathbb{S}^2$ surfaces) guarantee existence of solutions of 
(\ref{phi}) while large ones exclude existence. 
Finally, we apply Premoselli's theorem  \cite{BP}. Combined with auxiliary results
from bifurcation theory, with results on stability and symmetry collected in Sect 2., 
as well as with numerical  methods, we are able to clarify the bifurcation structure of the
axially symmetric solutions in terms of the bifurcation parameter $b = 3J\Lambda/2$: 
Firstly, there is a pair of ``principal'' branches consisting of stable and unstable
solutions all of which inherit the $O(2) \times O(2)$- symmetry of the SM. These branches
emanate at $J = 0$ from the solutions $\phi \equiv 0$ and $\phi \equiv 1$,
respectively, and meet at some marginally stable solution $\phi_*$ 
corresponding to a maximal angular momentum $J_*$. Moreover, off certain points on the unstable principal branch 
there bifurcate  branches which break the O(2)-symmetry along the
$\mathbb{S}$ direction, and which terminate at the Kottler solutions in the limit of 
vanishing angular momentum. We summarize these facts as Conjecture 1, which
also includes the hypothesis that there are no solutions which break the axial symmetry on $\mathbb{S}^2$. 
    
In Sect. 3.4.  we consider as SM the standard maximal slice of Kerr-de Sitter. 
For any fixed $\Lambda$, we take a family of $(t,\varphi)$-symmetric 
data generated by $\omega(r, \theta,\Lambda, J, a, m)  = J \omega_K(r, \theta,\Lambda, a, m)/J_{K}$ 
where $\theta$ and $r$ are ``Boyer-Lindquist'' coordinates, $J$ is the
``true''  angular momentum and $\omega_K$ is the twist potential generating Kerr-de Sitter with angular momentum 
$J_{K} = ma(1 + \Lambda a^2/3)^{-2}$ in terms of its standard parameters $m$ and $a$.  
This example is particularly well suited to illustrate Premoselli's result \cite{BP}: 
Choosing $b = 3 J \Lambda/2$ as above, it trivially implies existence of solutions to (\ref{con}) for all $J \le
J_{K}$. More interestingly, it also shows that Kerr-de Sitter can be ``overspun'' in the sense
that there exist data with $J > J_{K}$ as long as they remain strictly unstable. 
The latter is guaranteed in particular by the criterium $\wt \Omega^2 < \Lambda$ of Sect. 2.1.
mentioned above, but in any case for sufficiently small  $J_{K}$ and small $J - J_{K}$ 
(again compared to $\Lambda^{-1}$). These facts are collected in Theorem 4.  

\section{Stability, symmetry, existence and non-existence}

\subsection{Stability}

In the following discussion we refer to SMs and IDs 
as defined in Defs. 1. and 2.. As a rule the metric and the second fundamental form of
IDs will carry tildes. We note, however, that the SM in our examples  
(Sects 3.3. and 3.4.) trivially satisfy the constraints as well, 
so they are IDs on their own. Hence the task here is actually to generate
non-trivial IDs from trivial ones. 

We first define and discuss here (linearisation-) stability of solutions of (\ref{phi}) 
under conformal deformations, which will be crucial in the following results.
The linearized operator $L_{\phi}$ corresponding to (\ref{phi}) applied to
some function $\gamma$ reads

\begin{equation}
\label{lin}
L_{\phi} \gamma := - \left( \Delta - \frac{R}{8} \right) \gamma - \frac{5 \Lambda}{4}
\phi^4 \gamma + \frac{ 7 \Omega^2}{8 \phi^{8}} \gamma 
\end{equation}

\newpage
\begin{defn}~
\label{stab}
\begin{enumerate}
\item
A {\bf solution $\mathbf{ \phi}$} of (\ref{phi}) on a SM $({\cal M},g_{ij}, K_{ij})$
is called strictly stable, stable, marginally stable,
unstable or strictly unstable if the lowest eigenvalue $\varsigma$ in
(\ref{lin}) at $\phi$ satisfies  $\varsigma > 0$, 
 $\varsigma \ge 0$, $\varsigma = 0$,  $\varsigma \le 0$, or $\varsigma < 0$,
respectively. 
\item {\bf ID} given by Equ. (\ref{con}) are said to have lowest eigenvalue
$\varsigma$ (under conformal deformations) if (\ref{lin}) has this lowest eigenvalue  at $\phi \equiv 1$. 
The ID are called strictly stable if  $\phi \equiv 1$ 
is strictly stable, and analogous definitions for ID apply with the other stability
properties.

\end{enumerate}
\end{defn}

The natural question raised by these definitions is resolved as follows.

\begin{lem} 
\label{stabil}
A strictly stable solution $\phi$ of (\ref{phi}) on a SM
$({\cal M},g_{ij}, K_{ij})$ defines via (\ref{ct}) strictly stable ID
 $(\wt{\cal M},\wt g_{ij}, \wt K_{ij})$. The same applies to the other
stability properties of Definition \ref{stab}. 
\end{lem}

\begin{proof}

We show, more generally, that only the conformal class of the SM matters for
stability of the solution of (\ref{phi}) and for the resulting ID.
We first note that (\ref{phi}) is obviously conformally invariant 
in the sense that when $\phi$ solves (\ref{phi})
and defines ID $(\wt {\cal M},\wt g_{ij} = \phi^4 g_{ij}, \wt K_{ij} =
\phi^{-2} K_{ij})$, then  $\wh \phi = \vartheta^{-1} \phi$ solves (\ref{phi})
on the SM  $(\wh {\cal M},\wh g_{ij} = \vartheta^4 g_{ij}, \wh K_{ij} =
\vartheta^{-2} K_{ij})$ and defines the same ID.
A conformal covariance property also holds for the linearisation {\it operator} 
(\ref{lin}) in the sense that the rescaling $\wh \gamma  =
\vartheta^{-1}\gamma$ gives  $\wh L_{\wh \phi}  \wh \gamma = \vartheta^{-5} L_{\phi}
\gamma$ in terms of $\wh L_{\wh \phi}$ on $\wh {\cal M}$. 
A subtlety now arises since the eigenvalue equation  $L_{\phi} \mu= \lambda \mu$
 is obviously {\it not conformally invariant} 
(when the eigenfunction $\mu$ is scaled as above)
and the same applies to the eigenvalues themselves.
However, what matters for stability is only the sign (or the vanishing) of 
the lowest eigenvalue $\varsigma$. To show that this is actually invariant we 
recall the Rayleigh-Ritz characterisation,
\begin{equation}
\varsigma = \inf_{\gamma \in C^{\infty}, \gamma \not\equiv 0} \frac{\int_{\cal M} \gamma
L_{\phi} \gamma dv}{\int_{\cal M} \gamma^2}
\end{equation}
and note that its numerator is invariant, while the denominator is
manifestly positive. The statement of the Lemma is now obtained by setting  
$\vartheta = \phi$ in the above arguments. 
\end{proof}
\noindent
{\bf Remarks.}
\begin{enumerate}
\item Recalling that the eigenvalues $\lambda$ depend on the conformal scaling of the metric in general, 
we denote by $\wt \lambda$ the eigenvalues w.r.t to 
the generated ID (i.e. when $\phi \equiv 1$ in (\ref{lin})).
\item
The above definitions of stability under conformal deformations have nothing to do with 
dynamical stability of the solutions evolving from the data. We will return to this issue 
in connection with the Kerr-de Sitter example in Sect.3.4.  
\end{enumerate}
\begin{prop}
\label{stest} 
Let $(\wt {\cal M}, \wt g_{ij}, \wt K_{ij})$  be ID with volume $\wt V$ and lowest eigenvalue
$\wt \varsigma$. Then 
\begin{enumerate}
\item
\begin{equation}
\label{OgV}
\int_{\wt {\cal M}} \wt \Omega^2 d \wt v \ge (\Lambda + \wt \varsigma) \wt V.
\end{equation}
 \item
If $\wt \Omega^2 \le \Lambda$ and  $\wt \Omega^2 \not\equiv \Lambda$ on $\wt {\cal M}$, then the ID are strictly unstable. 
\item If $\wt \Omega \equiv 0$, then $\wt \varsigma = - \Lambda$, while $\wt
\Omega^2 \equiv \Lambda$ implies $\wt \varsigma = 0$.
\end{enumerate}  
\end{prop}

\begin{proof}
Combining (\ref{con}) with (\ref{lin}) for $\phi \equiv 1$ we find 
\begin{equation}
\label{Lz}
\wt \Delta  \wt \zeta = (\wt \Omega^2 - \Lambda - \wt \varsigma)\wt \zeta
\end{equation}

 where $\wt \zeta$ is the eigenfunction corresponding 
to the lowest eigenvalue $\wt \varsigma$.
As $\wt \zeta$ has no zeros it can be chosen to be positive.
Dividing by $\wt \zeta$ and integrating, we obtain
\begin{equation}
\label{OeV}
\int_{\wt {\cal M}} (\wt \Omega^2 -  \wt \zeta^{-2} |\wt \nabla \wt \zeta|^2)
d\wt v   =  
(\Lambda + \wt \varsigma) \wt V
\end{equation}
which implies assertions 1. and 2.. Point 3. is a consequence of
the maximum principle applied to (\ref{Lz}).
\end{proof}
{\bf Remark.} We wish to clarify here an important point which arises by
combining  Lemma \ref{stabil} with Proposition \ref{stest}: 
These results {\it do not imply} that any solution of (\ref{lin}) with
$\Omega$ small enough on a SM ${\cal M}$ leads to unstable ID $\wt {\cal M}$. Rather, the tilde 
on $\wt \Omega^2$ in the requirement $\wt \Omega^2 = \phi^{-12} \Omega^2 \le  \Lambda$ of 
Proposition \ref{stest} must not be overlooked.  In fact, the stable examples with small 
$\Omega$ but large $\wt \Omega$ (and hence small $\phi$)  will play a key role below.

\subsection{Stability and symmetry}

We recall from the introduction that ``symmetry-inheriting''  and ``symmetry
breaking'' are the properties of solutions of (\ref{phi}) of (not) sharing all 
symmetries of the equation.
This behaviour is related to  (in-)stability of solutions; we will observe
it in the Bowen-York example of Sect. 3.3.3.  We note here the following 

\begin{prop}
\label{csym} 
 Assume that a SM $({\cal M}, g_{ij})$ and $\Omega$ have a continuous symmetry $\xi$, i.e. 
\begin{equation}
{\cal L}_{\xi} g_{ij} = 0 \qquad {\cal L}_{\xi} \Omega = 0.
\end{equation}  
where ${\cal L}_{\xi}$ is the Lie derivative.
Then all stable solutions of (\ref{phi}) are invariant as well, i.e.
\begin{equation}
{\cal L}_{\xi} \phi = 0.
\end{equation} 
\end{prop}

\begin{proof} 
We first note that, for any solution $\phi$ of (\ref{phi}),
$ {\cal L}_{\xi} \phi$ is a solution of the linearized equation: 
Using that the Lie derivative ${\cal L}_{\xi}$ 
commutes with $\Delta$ and $R$ for invariant metrics $g_{ij}$ 
gives the second equation in 
\begin{equation}
\label{Lie}
 0 = {\cal L}_{\xi} \left[ \left( \Delta - \frac{R}{8} \right) \phi + \frac{\Lambda}{4}
 \phi^5 + \frac{\Omega^2}{8 \phi^{7}} \right] = L ({\cal L}_{\xi} \phi)
\end{equation}  
while the first one is obvious from the fact that $\phi$ solves
(\ref{phi}).

Next, as $\phi$ is stable, the lowest eigenvalue of $L$ is non-negative. 
As already noted in Sect. 2.1., the lowest eigenvalue is always non-degenerate and the corresponding 
eigenfunction has no zeros. 
It follows from (\ref{Lie}) that either ${\cal L}_{\xi} \phi \equiv 0$ which we want to prove, 
or that ${\cal L}_{\xi} \phi$ is a ground state eigenfunction with eigenvalue zero
and without zeros; by changing sign of $\phi$ if necessary, we thus have 
\begin{equation}
\label{Liephi}
{\cal L}_{\xi} \phi > 0.
\end{equation}
This we rule out as follows.
Since $\xi^i$ is a Killing vector, (\ref{Liephi}) can be rewritten as
\begin{equation}
\nabla_i (\phi \xi^{i}) > 0.
\end{equation}
But as ${\cal M}$ is compact, the l.h.s integrates to zero and gives a
contradiction. 
\end{proof}

We remark that arguments along the lines above can be and have been 
applied to a  large class of semilinear and quasilinear elliptic equations on compact 
manifolds (cf. e.g. Sect.8 of \cite{AMS}).

\subsection{Existence, non-existence and stability}

We adapt here the results of Hebey et al. \cite{EH,HPP,HV} and Premoselli \cite{BP} to the present
context in order to obtain bounds on $\Omega$ and its integrals in terms of geometric quantitites, which
guarantee existence or non-existence of solutions of (\ref{phi}).   
Premoselli's result also has an impact on the relation between  stability
and symmetry, which we state in Corollary 1. after the Theorem. 

We remark that the results \cite{HPP,BP} refer to a more general equation
than (\ref{phi}) in which $R$ and $\Lambda$ can be replaced by a large class
of functions. A feature of the present equation (\ref{phi}) already noted in
Sect. 2.1. is its conformal invariance which is obvious from the purpose
which is serves, and which  simplifies the (non-)existence criteria.

Turning now to the results of \cite{HPP}, we first note that the existence result Thm 3.1 is
indeed formulated in an invariant way under conformal rescalings
$\wh g_{ij} = \vartheta^4 g_{ij}$ of the SM provided the test function $\varphi$ is assigned  the 
conformal weight $\wh \varphi = \vartheta^{-1} \varphi$.
On the other hand, the non-existence result, Thm.2.1. of \cite{HPP} is not
conformally invariant.
 We reproduce these results below (Thm. \ref{t1}) under the simplifying assumption 
that the SM has constant curvature. Needless to say, this restriction breaks
 conformal invariance. In the examples discussed below, the SM $\mathbb{S}^2 \times \mathbb{S}$ of Sect.
 3.3. has constant curvature $R = 2\Lambda$
while the Kerr de Sitter data of Sect 3.4. have not. 
  
In the existence criterium (\ref{ex}), there enters the Yamabe constant
\begin{equation}
\label{Yam}
Y = \inf_{\gamma \in C^{\infty}, \gamma \not\equiv 0}  
\frac{\int_{\cal M} \left(8 |\nabla \gamma|^2 + R \gamma^2 \right) dv}
{ \left( \int_{\cal M} \gamma^6  dv  \right) ^{1/3} }.
\end{equation}
%where $\nabla$ is the covariant derivative w.r.t. $g_{ij}$.
We note that $Y$ is conformally invariant. Moreover, by virtue of Yamabe's theorem \cite{LP}, there
always exists a scaling which minimizes $Y$, and such a minimizer has constant curvature.
Therefore, the definition (\ref{Yam}) can be reduced to  $Y = \inf (R V^{2/3})$  
where $V$ is the volume of ${\cal M}$, and the infimum is taken over all metrics with 
$R= const.$ within the conformal class.

We also remark that, in order to have a chance of satisfying (\ref{con}) it is clear 
that we have to set out from SM $({\cal M}, g_{ij})$ which are in the positive Yamabe class, i.e. $Y > 0$.

\noindent
\begin{thm} 
\label{t1}
We take a SM $({\cal M},g_{ij}, K_{ij})$  as defined in Def.
(\ref{seed}) but
with constant scalar curvature $R$.    
 \begin{enumerate}
\item Assuming that 
\begin{equation}
\label{ex}
\int_{\cal M} \Omega^{2} dv \le \frac{Y^6}{256 \Lambda^2 R^3 V^3}  
\end{equation}
Equ. (\ref{phi}) has a smooth, positive solution.
\item
Assume that
\begin{equation}
\label{noex}
\int_{\cal M} \Omega^{5/6} dv >   \frac{R^{5/4} V}{3^{5/4} \Lambda^{5/6}} 
\end{equation}
then (\ref{phi}) has no smooth, positive solution.
\end{enumerate}
\end{thm}

\begin{proof}

For $\Omega \equiv 0$ the first part is obvious from  Yamabe's theorem \cite{LP}.
Otherwise, this part is a direct application of Thm. 3.1. of \cite{HPP},
observing that the Sobolev constant $S_h$ of this theorem is related to the Yamabe
constant $Y$ via $ (Y/8)^3 = 1/S_h$ in the present situation, and setting the ``test function''
$\varphi \equiv 1$. 
(Thereby we are likely to miss the optimal numerical factor on the r.h.s. of (\ref{ex})).
The second part follows readily from Thm. 2.1. of
\cite{HPP} except for the fact that $\Omega > 0$ was required there. 
The extension which allows for zeros in $\Omega$ is covered by Thm. 3  of
\cite{EH}.
\end{proof}

We now rewrite Premoselli's results \cite{BP}.

\begin{thm}
\label{bp}
We decompose $\Omega = b \Omega_0$  in (\ref{phi}) (in a non-unique way) 
in terms of a constant $b > 0$ and a function $\Omega_0$. The following
statements  refer to the  solubility of (\ref{phi}) on a SM depending on the choice of $b$, 
when $\Omega_0$ is kept fixed: There exists $0< b_* < \infty$ such that
(\ref{phi}) has
\begin{enumerate}
\item At least two  positive solutions for $b < b_*$, at least one of
which is strictly stable. Moreover, one of the strictly stable solutions, 
called  $\phi(b)$, is ``minimal'' in the sense that for any positive solution $\phi \not\equiv
\phi(b)$ we have $\phi > \phi(b)$.   
\item A unique, positive solution for $b = b_*$ which is marginally stable.
\item No solution for $b > b_*$
\end{enumerate}
\end{thm}
\begin{proof}
These statements just combine Theorem 1.1., Proposition 3.1. (positivity of solutions) 
and Proposition 6.1. (stability) of \cite{BP}. (The statement of strict stability in point 1. is not
explicit in the formulation of the latter Proposition, but contained in its  proof).    
\end{proof}

The following extension of Proposition \ref{csym}  is an immediate consequence of this
theorem.

\begin{cor} 
\label{dsym}
If $({\cal M}, g_{ij})$ and $\Omega$ have a discrete symmetry, then the stable ``minimal'' solution of 
point 1. in Premoselli's theorem, 
as well as the unique solution of point 2. of this theorem share this symmetry.
\end{cor}

As typical for non-linear equations, we expect bifurcations to occur among
the set of solutions of (\ref{phi}).  While the detailed behaviour of this
set will depend on $g_{ij}$, $\Lambda$ and $\Omega$,  Premoselli's theorem indicates that
 $b$ plays a distinguished role as bifurcation parameter. The key values of $b$
 for  understanding the structure of the solutions are $b = 0$, and the ``critical'' values 
by which we mean those for which the linearised operator (\ref{lin}) has a
 non-trivial kernel. The latter is the case in particular at $b = b_*$, 
but in general (and in particular in the Bowen-York example in Sect. 3.3.3.) 
more such critical values will show up.

We first discuss $b = 0$. While  Premoselli's theorem
does not apply to this case, it is known that Equ. (\ref{phi})
has at least two  solutions:
\begin{description}
\item[$\mathbf {\phi > 0}$:] 
We first mention the special case $R= 2\Lambda$ where there is the trivial solution $\phi \equiv
1$; however, there are many more (Kottler-) solutions which we revisit in Sect.3.3.1. 
In the general case, Yamabe's theorem mentioned above  guarantees the existence of
at least one positive solution. We now observe that {\it all regular solutions are necessarily unstable} 
in the sense of Def. \ref{stab}; this follows from point 3. of Proposition
\ref{stest}, while point 2. shows that instability still holds for small 
$\wt \Omega = b \phi^{-6} \Omega_0$ and therefore small $b$. However, in
this context it is important to avoid an instructive catch:  
Premoselli's theorem asserts the {\it existence of at least one stable solution for all small enough  $b$} 
(and therefore, for small enough $\Omega = b \Omega_0$).   The key to resolving  this issue  is the same as in the 
Remark after Proposition \ref{stest}, namely a tilde: $\wt \Omega = \phi^{-6} \Omega = b
\phi^{-6} \Omega_0$. We are led to the conclusion that the conformal
factor $\phi$ which generates the stable branch of solutions from any SM 
must go to zero when $b \rightarrow 0$ in order to allow $\wt \Omega$ to violate the instability
condition $\wt \Omega^2 < \Lambda$ (or its integral). This leads us to the other solution of
(\ref{phi}) for $b = 0$, namely
\item[$\mathbf {\phi \equiv 0}$:] 
While useless as conformal factor, the above arguments indicate that this solution is the 
origin of the unique ``minimal'' branch of Premoselli's theorem. Proposition \ref{b0} 
confirms and clarifies this.  
\end{description}

In the following result the rescaling $\psi = b^{-1/4} \phi$
will be crucial. In terms of this variable, we obtain from (\ref{phi}) 
\begin{equation}
\label{psi}
\Psi(\psi) := - \left( \Delta - \frac{1}{8} R \right) \psi - \frac{b}{4}\Lambda \psi^5
- \frac{\Omega_0^2}{8 \psi^{7}} = 0.  
 \end{equation}
More precisely, (\ref{psi}) is equivalent to (\ref{phi}) only
 for $b > 0$, but we consider the former equation  for $b \ge 0$.
Note that the constant $b$ which controlled the size of the momentum term in (\ref{phi}) 
now scales the cosmological constant in (\ref{psi}).

 We also introduce the linearisation at some $\psi$, 
\begin{equation}
\label{linpsi}
L_{\psi} \gamma = - \left( \Delta - \frac{1}{8} R \right) \gamma - \frac{5 b}{4}\Lambda
\psi^4 \gamma + \frac{7 \Omega_0^2}{8 \psi^{8}} \gamma. 
\end{equation}

\begin{prop}
\label{b0}
For sufficiently small $b \ge 0$, the equation (\ref{psi})
 on a given SM $({\cal M},g_{ij}, K_{ij})$ has a unique, positive, strictly stable solution $\psi(b)$.
\end{prop}

\begin{proof}
For $b = 0$, it can be shown via the sub- and supersolution method \cite{JI} that
the resulting Lichnerowicz equation (\ref{psi}) has a unique, positive
 solution $\psi_0$. Next, strict stability follows readily from the linearisation  
\begin{equation}
\label{lin0}
L_0 \gamma = - \left( \Delta - \frac{1}{8} R \right) \gamma + \frac{7 \Omega_0^2}{8 \psi_0^{8}}
\gamma. 
\end{equation}
In particular, using that the Yamabe constant of ${\cal M}$ is positive, $L_0$ has a trivial kernel. 
This allows application of the implicit function theorem and indeed yields the desired conclusion.
\end{proof}

{\bf Remark.} For $b \ge 0$ we clearly recover here the beginning of the unique strictly 
stable minimal branch of solutions from point 1. of Premoselli's theorem.
Note that, for $b \rightarrow  0$, regularity of $\psi$ indeed entails  
$\phi \rightarrow 0$, as anticipated in the discussion above and in the
Remark after Proposition 1.   
Bounds on integral norms of $\phi$ in terms of $b$,$\lambda$, $\Lambda$ and
$\Omega_0$ can be derived via Equ. (\ref{OgV})  but will not be given here.

 We now turn to the critical values of $b$. 
As the analysis is slightly simpler in terms of the variable $\psi$ 
 compared to $\phi$, we continue working with (\ref{psi}) and its
 linearisation $(\ref{linpsi})$ rather than with 
the equivalent original Lichnerowicz equation (\ref{phi}).  

Here the only simple case  is the marginally stable (lowest eigenvalue zero) one, 
which arises in particular at the maximal value $b = b_*$. In this case a simple bifurcation
analysis leads to the following behaviour of the solutions:

\begin{prop}
\label{bstar}
Assume that $\psi_c$ is a marginally stable solution of (\ref{psi}) for some value $b_{c}$. 
Then there is a solution curve $(b(s), \psi(s))$   near $\psi_c$ which ``turns to the left'' 
(i.e. towards smaller values of $b$) at $(b_c, \psi_c)$ . This entails that there is
an $\epsilon > 0$ such that for all  $b \in (b_{\star} - \epsilon, b_{\star})$ there are precisely two solutions, 
at least one of which is strictly stable.   
\end{prop}

\begin{proof}
The requirements of the Crandall-Rabinowitz theorem in the form  Thm. 3.2. of \cite{CR} are as follows:
\begin{enumerate}
\item
The kernel $\upsilon_c$ of the linearisation $L_c$ defined in (\ref{linpsi}), 
and of its adjoint, are one-dimensional at the critical solution $\psi_c$. 
 \item
The derivative $d \Psi /db|_{b_{c}}$ of the Lichnerowicz operator (\ref{psi})
is not in the range of the linearised operator at $b_{c}$. 
\end{enumerate}
Now 1.  follows from the assumption of marginal stability and the fact that $L_c$ is self adjoint in
the present case. Proving 2. is equivalent to 
showing that 
 \begin{equation}
\label{Lgam}
L_c \gamma = \frac{1}{4} \psi_c^5
\end{equation}
has no solutions. Assuming the contrary and using the fact that the ground state eigenfunction 
$\upsilon_c$ can be chosen  to be positive, we indeed obtain the contradiction
\begin{equation} 
0 = \int_{\cal M} \gamma L_c \upsilon_c =  \int_{\cal M} \upsilon_c  L_c \gamma = \frac{1}{4}
\int_{\cal M}  \psi_c^5 \upsilon_c > 0. 
\end{equation}  
From the Theorem we conclude that near the bifurcation point $(b_c,\psi_c)$ there is a curve of solutions
$(b(s),\psi(s))=(b_c + \beta(s),\psi_c + s \upsilon_c + \tau(s))$ with $\beta(0)=\beta'(0)=
\tau(0)=\tau'(0)=0$. To compute $\beta''(0)$ we differentiate 
Eq.
(\ref{psi}) twice with respect to $s$ and evaluate at $s=0$
\begin{equation}
\label{psidp}
  L_c \psi_c''  - \left(7 \Omega_0^2 \psi_c^{-9} + 5 b_c \Lambda \psi_c^3 \right) \,
 \upsilon_c^2 - \frac{\Lambda \beta''(0)}{4} \upsilon_c^5 =0\,.
  \end{equation}
Multiplying this equation by $\upsilon_c$, Equ. $L_c \upsilon_c = 0$ by
 $\psi_c'' $ and subtracting we obtain
\begin{equation}
\label{beta}
  \beta''(0)=-\frac{4 \int_{\cal M} \left(7 \Lambda ^{-1} \Omega_0^2 \psi_c^{-9} + 5 b_c  \psi_c^3 \right) \, 
\upsilon_c^3\,}{ \int_{\cal M} \psi_c^5\, \upsilon_c } <0.
\end{equation}
Thus,  $(b_c,\psi_c)$ is the turning point at which the curve of solutions turns to the left.
\end{proof}
Regarding the behaviour of the solution curve near general critical values $b_c$
(i.e. with negative lowest eigenvalue), it depends largely on the precise form of the equation. 
We proceed with discussing examples.  

\section{$(t,\varphi)$-symmetric seed manifolds}

\subsection{Angular momentum}

\label{ang}

We recall here standard material on axial symmetry (AS) and on the angular momentum of
compact 2-surfaces of spherical topology.
 A SM  $({\cal M},g_{ij}, K_{ij})$ is AS iff
the circle group acts effectively on ${\cal M}$ and its set of fixed points is non-empty.
This implies the existence of a Killing field ${\eta_i}$ 
with fixed points along an axis, such that
\begin{equation}
{\cal L}_{\eta} g_{ij} = {\cal L}_{\eta} {K_{ij}} = 0. 
\end{equation}
The angular momentum $J$ of a compact 2-surface ${\cal S}$ in an AS SM 
is given by
\begin{equation}
\label{am}
J = \frac{1}{8\pi} \int_{\cal S} K_{ij} \eta^{i} dS^{j}.
\end{equation}
Since our Definition \ref{seed} of a SM contains the requirement that $K^{ij}$ is divergence-free, 
all homologous 2-surfaces have the same angular momentum. 
This implies that, in order for $J$ to be non-zero, the homology group $H_2({\cal M})$ must be non-trivial.
We also note that the above definition of $J$ is conformally invariant.

When the SM is AS, so are the stable solutions of (\ref{phi}) by Proposition \ref{csym} above. 
The same then applies to the ID and, by standard ADM  evolution, to the evolving spacetime 
$(\wt {\cal N}, \wt g_{\mu\nu})$ ~$(\mu,\nu = 0,1,2,3)$. Any AS spacetime satisfies
\begin{equation}
\label{Ein}
\wt R_{\mu\nu} = \Lambda \wt g_{\mu\nu} \qquad {\cal L}_{\wt \eta} \wt g_{\mu\nu} = 0
\end{equation}
where $\wt R_{\mu\nu}$ the Ricci tensor of $\wt g_{\mu\nu}$ and $\wt
\eta^{\mu}$ is the spacetime Killing vector.

The angular momentum and its properties can alternatively be discussed in terms of spacetime quantities. 
In particular, the definition (\ref{am}) now reads

\begin{equation}
\label{Komar}
J = \frac{1}{8\pi} \int_{\cal S} \widetilde \nabla^{\mu} \widetilde \eta^{\nu}
\widetilde{  d S_{\mu\nu}}
\end{equation}   

where $\wt \nabla^{\mu}$ denotes the covariant derivative of $\wt g_{\mu\nu}$, 
and $\wt {d S_{\mu\nu}}$ is the volume element of ${\cal S}$. 
From (\ref{Ein}), the integrand of (\ref{Kom}) satisfies 
\begin{equation}
\label{div}
\wt \nabla_{\mu} \left( \wt \nabla^{\mu} \wt \eta^{\nu} \right) = - \wt
R^{\nu}_{~\mu} \wt \eta^{\nu} = - \Lambda \wt \eta^{\nu}.
\end{equation}

We now recover the spacetime version of the invariance result for $J$:  By Gauss' theorem, (\ref{div}) 
implies that all 2-surfaces ${\cal S}$ which are homologous and bound an AS 3-surface 
have the same angular momentum.  

We next introduce the twist vector $\wt \omega_{\mu}$ of the Killing vector $\wt \eta^{\mu}$
\begin{equation}
\label{twi}
\widetilde \omega_{\mu} = \wt \epsilon_{\mu\nu\sigma\tau} \widetilde \eta^{\nu}
\widetilde \nabla^{\sigma} \widetilde \eta^{\tau}
\end{equation}   
where  $\wt \epsilon_{\mu\nu\sigma\tau}$ is totally antisymmetric 
and $\wt \epsilon_{0123} = \sqrt{|det \, \wt g_{\mu\nu}|}$.
 $\widetilde \omega_{\mu}$ is curl-free by virtue of (\ref{Ein}), i.e. $\widetilde \nabla_{[\mu} \widetilde \omega_{\nu]} =
0$. Hence there exists locally a twist potential  $\wt \omega$, defined up to a constant, 
such that $\wt \omega_{\mu} =  \widetilde \nabla_{\mu} \widetilde \omega$.

For an AS 2-surface $ {\cal S} \subset {\cal M}$ of spherical topology, 
the twist potential allows the following reformulation of the angular momentum
(equivalent to (\ref{am}) and (\ref{Komar})),
\begin{equation}
\label{J}
J = \frac{1}{8} \left[ \omega(N) - \omega(S) \right]
\end{equation}
where $N$ and $S$ are the poles of ${\cal S}$.

\subsection{$(t,\varphi)$-symmetric seed manifolds}

Bardeen \cite{JB} investigated data for rotating stars which, in terms of particle physics terminology, enjoy a PT-invariance, 
i.e. their evolution is invariant under the simultaneous change of time and spin  direction. 
Following Dain \cite{SD} and Dain et al. \cite{DLT} who systematically investigated such
seed manifolds we call them  $(t,\varphi)$-symmetric. This construction can be summarized as follows.

\begin{defn}
\label{tp1}

An AS SM $({\cal M},g_{ij}, K_{ij})$ 
is called $(t,\varphi)$-symmetric (TPSM) if 
\begin{enumerate}
\item
The axial Killing field $\eta_i$ is hypersurface orthogonal,
 i.e. $\epsilon_{ijk} \eta^i \nabla^j \eta^k = 0$ where 
 $\epsilon_{jkl}$ is totally antisymmetric
and $\epsilon_{123} = \sqrt{det \,g_{ij}}$.
\item 
$K^{ij}$ satisfies
\begin{equation}
K_{ij}\eta^i \eta^j = 0 ~~\mbox{and} ~~ K_{ij}q^{ik} q^{jl} = 0 
\end{equation}
where $q_{ij} = g_{ij} - \eta^{-1} \eta_{i} \eta_{j}$.
\end{enumerate}
\end{defn}

We now state a well-known result which yields an alternative formulation of a TPSM.  

\begin{prop}
\label{tp2}
~
\begin{enumerate}
\item
Let  $({\cal M},g_{ij},K_{ij})$ be a TPSM. 
Then there exists  a smooth scalar function $\omega$ such that
\begin{enumerate}
\item
 The axial Killing field $\eta_i$ leaves $\omega$ invariant, i.e.
 $\eta^l \nabla_l \omega = 0$, and 
\item The extrinsic curvature
\begin{equation}
\label{Kom}
 K^{ij}  = \frac{1}{\eta^2} \eta^{(i} \epsilon^{j) kl} \eta_k \nabla_{l} \omega   
\end{equation}
with $\eta =  \eta^k \eta_k$ is smooth everywhere, in particular on the axis.
\end{enumerate}
\item
Conversely, let $({\cal M},g_{ij})$ be a manifold of positive Yamabe type
such that
\begin{enumerate}
\item
$({\cal M}, g_{ij})$ is AS with hypersurface-orthogonal Killing vector $\eta^i$.
\item
There is a smooth function $\omega$ which satisfies 1.(a), and  $K_{ij}$ defined by (\ref{Kom}) 
satisfies 1.(b). 
\end{enumerate}
Then $({\cal M},g_{ij}, K_{ij})$ is a TPSM.  

\end{enumerate}

\end{prop}

\begin{proof}
Simple calculations and application of the Poincar\'e Lemma, cf.
\cite{SD,DLT}. 
\end{proof}

\newpage
{\bf Remarks.}~
\begin{enumerate}
\item 
In contrast to $\eta^i$ which is hypersurface orthogonal (w.r.t. a foliation of 2-surfaces) 
by definition of TPSM, the spacetime Killing vector $\widetilde \eta^{\alpha}$ which arises from
the corresponding data is no longer hypersurface orthogonal 
(w.r.t. a foliation of 3-surfaces) in general. 
\item
The twist potential $\wt \omega$ was defined in Sect. \ref{ang} for all 
AS spacetimes. 
If such spacetimes arise from TPSM generated by the scalar function $\omega$
via (\ref{Kom}), it can be shown that the restriction of $\widetilde \omega$ to 
the initial surface ${\cal M}$ coincides with $\omega$,  
provided of course that the respective additive constants are adapted. 
This justifies the synonymous notation.     

\item In a coordinate system $(\rho, \theta, \phi)$  where $\eta^i = \partial/\partial \varphi$
and the metric $g_{ij}$ is diagonal,  TPSM have
 $K_{\rho\varphi}$ and $K_{\theta \varphi}$ as  only non-vanishing components of
 $K_{ij}$.

\end{enumerate}

For a TPSM we  easily obtain from (\ref{Kom}) 
\begin{equation}
\label{Om}
\Omega^2 = K_{ij}K^{ij} = \frac{|\nabla \omega |^2}{2 \eta^2}.
\end{equation}

This is the key input for the Lichnerowicz equation in the 
subsequent applications.

From now on we specify the SM $({\cal M}, g_{ij})$ to have topology
$\mathbb{S}^2 \times \mathbb{S}$.
For all AS SM and ID of this topology, we recall from Sect.
(\ref{ang}) that the angular momentum does not depend on the selected 
$\mathbb{S}^2$-surface; one can therefore use the term ``angular momentum of the SM (ID)'' 
instead of  $\mathbb{S}^2$.

\subsection{The round hypertorus.}

Here we restrict ourselves to the metric 
\begin{equation}
\label{Nar}
ds^2 = \frac{1}{\Lambda} \left( d\alpha^2 + d \theta^2 + \sin^2 \theta d\varphi^2 \right)  
\end{equation}
where $\alpha \in [0, T]$ ``goes around'' the $\mathbb{S}$-direction. This metric obviously has $O(2) \times O(3)$ 
as isometry group.
Equ. (\ref{Nar}) is also the induced metric on a time symmetric slice of the Nariai spacetime \cite{HN}.
We now consider possible choices for $\omega$ and $\Omega$ on this background. 

\subsubsection{$\Omega \equiv 0$} 

We first recall from Sect. 2.3. that, on arbitrary backgrounds $({\cal M}, g_{ij})$, 
the Lichnerowicz equation (\ref{phi}) reduces for $\Omega \equiv 0$ to the Yamabe
equation. This is the equation which minimizes the Yamabe functional
(\ref{Yam}), and the corresponding solutions determine data with constant scalar curvature 
$ \wt R = 2 \Lambda$. 
For the present background (\ref{Nar}) with $R = 2\Lambda$,
the Yamabe problem  has been studied thoroughly \cite{RC,HVA,RS,CS})
 due to its simplicity, but also due to interesting degeneracy properties. 
We review the results here.  
 
While $\phi = \phi_0 \equiv 1$ is of course a solution of (\ref{phi}), there are
in addition $k$ positive solutions  $\phi_j$ iff $T \in (2\pi k, 2\pi (k+1)]$ ($k \in \mathbb{N}_0, 0 \le j \le
k$). For $j \ge 1$ these solutions are periodic in $\alpha$ with periods
$P(m,\Lambda) =  T/j$ (cf. (\ref{per}) below). 
The resulting metrics $\phi_j^4 g_{ik}$ are known as time-symmetric 
data for the Kottler (Schwarzschild-de Sitter)  spacetime which contain $j$
pairs of horizons (each pair consisting of a ``cosmological'' and a ``black hole'' one).
Explicitly, $\phi$ can be determined from the standard form of the
1-parameter family of time-symmetric Kottler data via 

\begin{equation}
\label{Kot}
 \wt ds^2_K  =
r^2 ( \frac{dr^2}{\sigma} + d\theta^2 +
 \sin^2 \theta d\varphi^2 )
 = \underbrace{r(\alpha ,m, \Lambda)^2}_{\phi^4/\Lambda} ( d \alpha^2 +
 d\theta^2 + \sin^2 \theta).
\end{equation}

Here $m \in \mathbb{R}$, $m < 1/(3\sqrt{3})$, $\sigma = (r^2 - 2mr - \Lambda r^4/3)$, 
and the horizons $r_b$ and $r_c$ are located at the positive zeros of $\sigma$. 
Obviously, the coordinates are related  via
\begin{equation}
\label{alp}
\frac{d \alpha}{dr} = \sigma^{-1/2}   \qquad  r(\alpha = 0) = r_b.
\end{equation}
Integrating over the circle and requiring that the period $P(m,\Lambda)$ 
fits on the hypertorus gives 
\begin{equation} 
\label{per}
P(m,\Lambda) := 2 \int_{r_b}^{r_{c}} \sigma^{-1/2}dr  = T/j \qquad j \in \mathbb{N}.
\end{equation}

This way the parameter $m$ acquires a dependence on $T/j$, and therefore the same applies to $\phi = \phi_j$. 
    
All such data have Ricci scalar  $\wt R = 2\Lambda$ but the manifolds  $(\wt {\cal M},
\phi_j^4 g_{ik})$ have different volumes, and the Yamabe constant $Y = \inf (\wt R \wt V^{2/3}) = 
2 \Lambda  \min_j \wt V_j ^{2/3}$ is ``realized'' by the manifold of minimal volume. 
For $T \le 2\pi$,  this is necessarily (\ref{Nar}) as $\phi = \phi_0 \equiv 1$ is the only solution of $(\ref{phi})$,
while for all $T > 2\pi$ the metric $\phi_j^4 g_{ik}$ with minimal volume
always turns out  to be $\phi_1^4 g_{ik}$. 

Note that the Lichnerowicz equation is independent of $\alpha$ while all  solutions except for $\phi \equiv 1$ do depend on it. 
In other words, we have here a simple example of  ``symmetry breaking''. In terms of the stability classification 
Def.\ref{stab}. we find that both (\ref{Nar}) as well as all Kottler data are unstable, 
in consistency with Propositions \ref{stest}. and \ref{csym}.  As already mentioned in Sect. 2.1.
the stability classification should be considered as a mathematical tool 
rather than interpreted physically. 

\subsubsection{$\Omega = const.$} 

An exhausitve analysis of the solutions of the Lichnerowicz equation 
in this case has recently been obtained by Chru\'sciel and Gicquaud \cite{CG}.
In particular, it has been pointed out in Thm 3.1. there that
the results of \cite{JLX} imply that all solutions of (\ref{phi}) are $O(3)$-symmetric, i.e. they only depend on $\alpha$.

The assumption $\Omega = const.$ leads to interesting problems regarding the bifurcation 
structure of solutions. However, it is incompatible with the $(t,\varphi)$-symmetric scheme described
in Definition \ref{tp1} on which we focus in this work.
To see this, we integrate  (\ref{Om}) with the present assumptions which
gives $\omega = 4J(\sin 2\theta - 2 \theta)/\pi$, $\Omega^2 = 128 J^2 \Lambda^3/\pi^2$
and  via (\ref{Kom}) a second fundamental form whose only non-vanishing
component is
\begin{equation}
K_{\alpha\varphi} = \frac{8 J \Lambda^{1/2} \sin \theta}{\pi}.
\end{equation}
This tensor is singular on the axis, however. 
We therefore consider different choices of $\omega$. 
  
\subsubsection{Bowen-York data} 

The standard setting for Bowen-York data is a flat SM $({\cal M}, g_{ij}, K_{ij})$ with 

\begin{equation}
\label{BY0}
K_{ij} = \frac{6}{r^3} \epsilon_{kl(i} J^k n^l n_{j)} 
\end{equation}
where $n^i$ is a radial unit vector, $r$ is the radius on $\mathbb{R}^3$,
and the  constant angular momentum  vector $J_i = J \delta_{zi}$ points in the
Cartesian $z$ direction \cite{RB,BY}.
Alternatively this second fundamental form can be constructed via (\ref{Kom}) from the
AS function  
\begin{equation}
\label{BY1}
\omega = - 2J (\cos^3 \theta - 3 \cos\theta)
\end{equation}
where $J$ is in fact the standard angular momentum, as follows from (\ref{J}).
We now carry these definitions over to the SM (\ref{Nar}) on $\mathbb{S}^2
\times \mathbb{S}$, replacing (\ref{BY0}) by
\begin{equation}
\label{BY2}
K_{ij} = 6 \Lambda^{\frac{3}{2}} \epsilon_{kl(i} J^k n^l n_{j)}. 
\end{equation}
Here the radial unit vector $n^i$ is orthogonal to $\mathbb{S}^2$ 
(pointing in the  ``$\alpha$-direction'' in the coordinates (\ref{Nar})),  while 
the  vector $J^i$ now reads $\Lambda^{1/2} J (\cos \theta \sin \theta, 0)$.
In terms of the construction described in Proposition \ref{tp2}., the
generating function still reads as above, namely (\ref{BY1}).

Inserting in (\ref{Kom}) shows that the only non-vanishing component of the second fundamental form is 
\begin{equation}
\label{K}
K_{\alpha\varphi} = 3 J \Lambda^{1/2} \sin^2 \theta.
\end{equation}
This tensor is indeed smooth on the axis, for the same reasons which yield
smoothness of the metric (\ref{Nar}). 

It now follows from (\ref{K}) that
\begin{equation}
\label{OBY}
\Omega^2  = K_{ij}K^{ij} =  \frac{|\nabla \omega|^2}{2 \eta^2} =  18 J^2 \Lambda^3 \sin^2 \theta
= 8 b^2 \Lambda \sin^2 \theta 
\end{equation} 
where we have set $b = 3 J \Lambda/2$.

Inserting (\ref{Nar}) and (\ref{OBY}) into (\ref{phi})  we
are left with 

\begin{equation}
\label{LichBY}
\frac{\Phi(\phi)}{\Lambda} =  - \underbrace{ \frac{\partial^2 \phi}{\partial \alpha^2}}_{A \phi}
- \underbrace{\left({}^2\Delta  -\frac{1}{4}\right) \phi}_{B \phi} - 
\frac{1}{4}\phi^5 - \frac{b^2 \sin^2 \theta}{\phi^7} = 0  
\end{equation}
where ${}^2\Delta$ is the Laplacian on the round $\mathbb{S}^2$.
The corresponding linearised operator around some $\phi$ reads  

\begin{equation}
\label{LinBY}
\frac{L_{\phi}}{\Lambda}\gamma = - \left(A + B  + \frac{5}{4}\phi^4 - \frac{7 b^2 \sin^2
\theta}{\phi^8}\right)
\gamma. 
\end{equation}

We first adapt the (non-)existence result Theorem \ref{t1} to this example.
We obtain
   
\begin{thm} ~\\
\vspace*{-5mm} 
\begin{enumerate}
\item If $T \le 2\pi$ and  $ \Lambda |J|  < 0.05$, Equ. (\ref{LichBY}) has a smooth, positive  solution. 
The same applies if $T \ge 2\pi$ and $ \Lambda |J|  < 2.01/T^2$. 
\item If $ \Lambda |J| > 0.165$, Equ. (\ref{LichBY}) has no smooth, positive solution. 
\end{enumerate}
\end{thm}

\begin{proof}
From Theorem \ref{t1} via a lengthy calculation, which uses 
$Y = 8\pi^{4/3} (T/2\pi)^{2/3}$ for $T\le 2\pi$ and 
$Y \ge 8\pi^{4/3}$ for $T\ge 2\pi$ in point 1. See \cite{SP} for details. 
\end{proof}

We now apply Premoselli's theorem and our results on symmetry and stability to
Equs. (\ref{LichBY}) and (\ref{LinBY}). 
To obtain a complete picture of the bifurcation structure 
of the solutions we have to  resort partially to numerical methods. 
  We first determine the ``principal'' branches of solutions which only depend on $\theta$
and which are equatorially symmetric, i.e. we assume $\phi = \phi(\theta)$ and 
$\phi(\theta) = \phi(\pi - \theta)$. Due to the symmetries of the SM,
and by virtue of Proposition \ref{csym}, we know that this class will include the unique 
stable, minimal  branch whose existence is guaranteed by Thm. \ref{bp}.
In order to regularise this branch near $b=0$ we now adopt the substitution 
$\psi =b^{-1/4} \phi$ introduced already in (\ref{psi}).
This gives  
 \begin{equation}
\label{phipsi}
 B \phi
 +\frac{1}{4} \phi^5  + \frac{b^2 \sin^2 \theta}{\phi^{7}}=0,
 \hspace*{4cm} B \psi +  \frac{b}{4} \psi^5  + \frac{\sin^2 \theta}
 {\psi^{7}}  = 0.  
\end{equation}

\begin{figure}[h!]
\hspace*{-1cm}
\includegraphics[angle=0,totalheight=6.2cm]{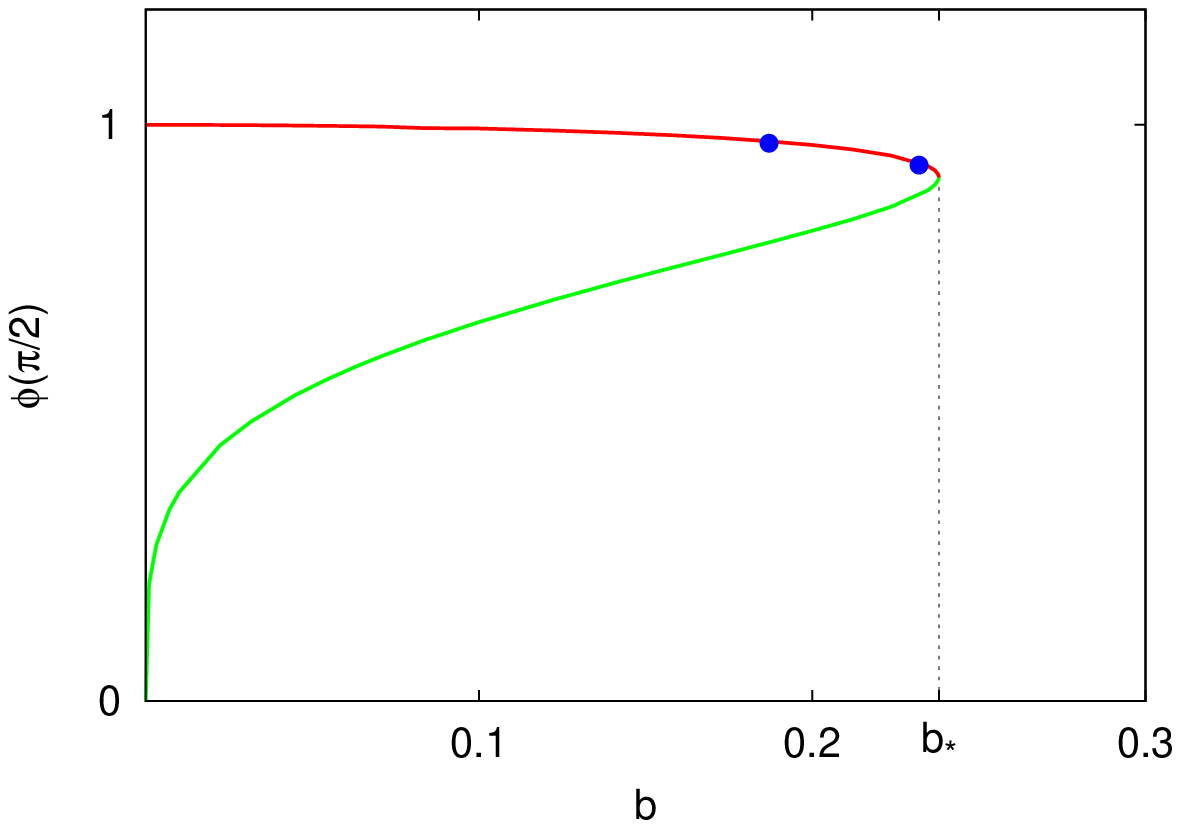}
%\hspace*{-0.2cm}
\includegraphics[angle=0,totalheight=6.2cm]{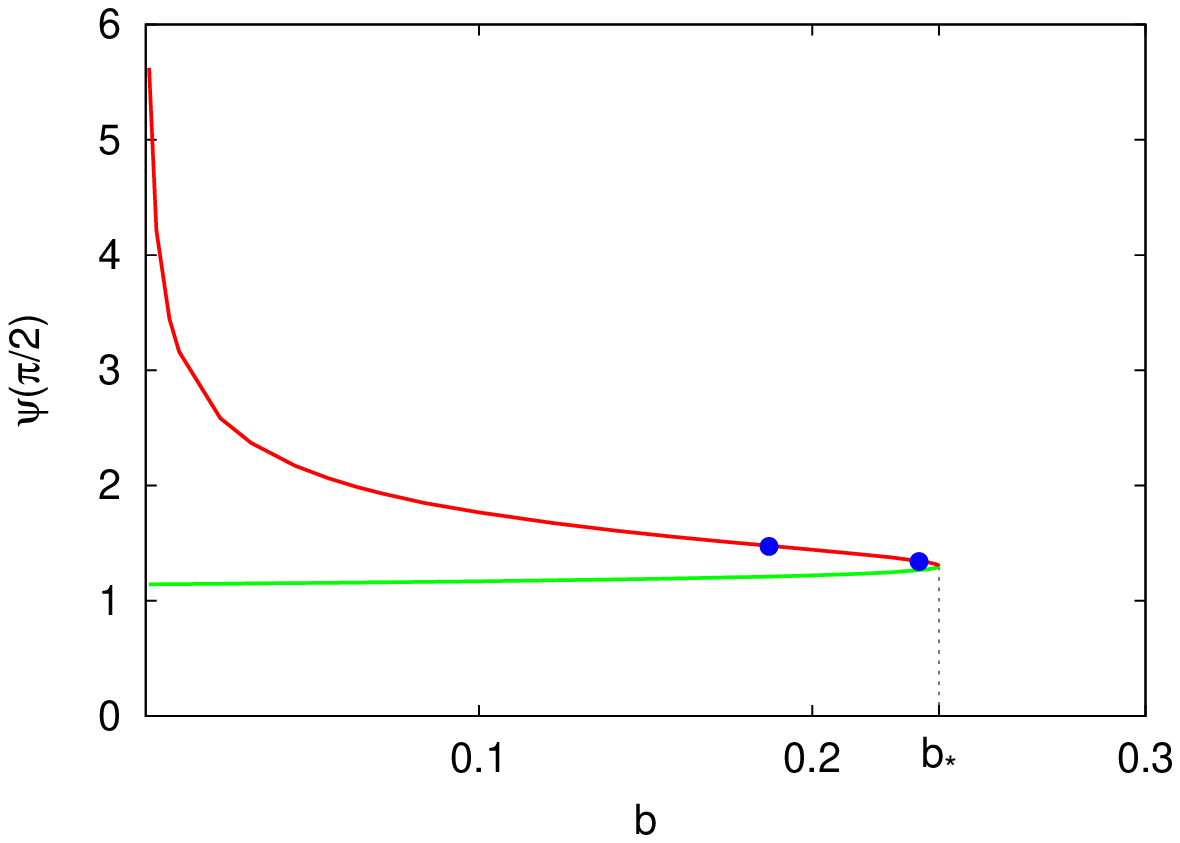}
\caption{The stable (green) and the unstable (red) axially
symmetric solutions of (\ref{phipsi}), (represented by their equatorial
values), and the bifurcation points of the secondary branches for the
sample period $T=5\pi$ (blue).}
\end{figure}

 As to obtaining the diagram on the left one first shows that, 
near a pole ($\theta = 0$) there is a 1-parameter family of analytic solutions of the form
$\phi(\theta) =  d + d(1 - d^4) \theta^2/16 + O(\theta^4)$.
One then ``shoots'' (numerically) such a local solution towards the equator and adjusts 
the parameter $d$ such that $\partial \phi/\partial \theta$ vanishes at the
equator ($\pi/2$).
This gives a stable branch emanating from $\phi \equiv 0$ (green in the online version), 
and an unstable branch (red) emanating from $\phi \equiv 1$ in consistency with Thm. \ref{bp}
and Proposition \ref{bstar}.
Moreover, numerics shows that the lowest eigenvalue $\varsigma$ of the linearisation ((\ref{LinBY}) with $A \equiv 0$) 
decreases monotonically from $0$ to $\varsigma = -\Lambda$ (cf. point 3. of Proposition
\ref{stest}) along the unstable branch. An analogous discussion 
in terms of $\psi$, which makes use of Proposition \ref{b0}, yields the diagram on the right.  

However, the ``principal'' branches displayed above cannot comprise all solutions - we recall from
Sect.3.3.1. that even in the special case $b = 0$ there is the symmetry-breaking Kottler family of data. 
We therefore now look for solutions of the form  $\phi(\alpha, \theta)$, i.e.
we still assume axial symmetry. For small $b$ one can in fact infer, via an implicit
function-argument, the existence of unstable ``secondary'' branches emanating from each Kottler solution.
To see where these branches end up, we rewrite an arbitrary eigenvalue
$\lambda$  of the linearised operator (\ref{LinBY}) at some solution $\phi(\theta)$ on the unstable principal branch 
in terms of the eigenvalue  $\lambda_{\theta}$ of the truncated operator ((\ref{LinBY}) for $A \equiv 0$). 
With a corresponding  separation of variables in the eigenfunction $\mu =  \mu(\alpha, \theta) = \mu_{\alpha}(\alpha)
\mu_{\theta}(\theta)$
we find
\begin{equation}
\label{lam}
\frac{\lambda}{\Lambda} = \left(\frac{2\pi}{T}\right)^2 j^2 +
\frac{\lambda_{\theta} } {\Lambda}  \qquad j= 0,1,2...
\end{equation}      
 Bifurcation theory (in essence the arguments of Proposition \ref{csym}) 
now shows that symmetry breaking can only occur at solutions where the linearised operator has a zero mode. 
Setting now  $\lambda = 0$ and
using the numerical observation that $\lambda_{\theta}$ changes
monotonically from $0$ at $b = b_*$ to some negative value (depending on $T$) at $b = 0$ we find 
that (\ref{lam}) has in fact $k$ solutions $j=1,2,....k$  when $T \in (2\pi k, 2\pi (k+1)]$
(Fig.1. shows the bifurcation points corresponding to $j=1$ and $j=2$ 
in the case $T = 5\pi$).
 Further numerical calculations (which are interesting
on their own and will be described in detail separately \cite{BMS})
now in fact reveal that from the bifurcation point labelled $j$  ($1\le j\le k$) on the 
unstable principal branch, there emanates a secondary branch of solutions of the form
$\phi_j(\alpha,\theta)$ which continues  till $b=0$ and terminates at the Kottler 
solution with precisely $j$ pairs of horizons (cosmological and black hole).

We remark that we cannot rule out the existence of non-axially symmetric
solutions. In particular, Thm. 3.1. of \cite{CG} has no obvious extension to the present case. 
However, the existence of such solutions is unlikely as we do not find corresponding bifurcation points on 
any of the known branches.
 
We summarize the above exposition in the following Conjecture. 
The status of point 3. is ``truly conjectural'' in the sense just mentioned, 
while  points 1. and 2. are ``facts''.
However, we wish to reserve the term ``theorem'' to a forthcoming
publication \cite{BMS} where we hope to present a complete analytic proof of
existence of the solutions, but in any case a detailed numerical analysis.   

\begin{conj}
The smooth, positive solutions of
the Lichnerowicz equation (\ref{LichBY}) have the following properties:
\begin{enumerate}
\item  There exists $b_*$ ($b_* \approx 0.238$) such that for $b \in (0,b_*)$ there
is precisely one stable solution and one unstable solution which only depend on $\theta$ 
and are equatorially symmetric (the stable and the unstable principal
branch).  Moreover, for $ b = b_*$, there is a unique marginally stable solution
with the same symmetries, while for $b > b_*$ there is no solution.
In the limit $b \rightarrow 0$, the stable solutions tend to zero like  $b^{1/4}$.
\item 
If $T \in (2\pi k, 2\pi (k+1)]$, there exist $k$ values $b_1$...$b_k$, with $b_* > b_1 > b_2 >...> b_k > 0$ such
that, from each point on the unstable principal branch corresponding to $b_j$,
there bifurcates a branch of unstable solutions depending only on $\alpha$ and $\theta$
(secondary branches). Each such branch continues till $b = 0$, where its end point represents the
Kottler solution with $j$ pairs of horizons. 
\item All solutions are axially symmetric (i.e. independent of $\varphi$)
\end{enumerate}      
\end{conj}

We finally note that the maximal value $b_* \approx 0.238$ 
obtained from numerics (cf. Fig. 1.), corresponding to $\Lambda |J| \approx
0.159$, significantly exceeds the existence limit given in Thm. 3.1., 
while it comes remarkably close to the non-existence limit, Thm. 3.2..

\subsection{Kerr-de Sitter}

The Kerr-de Sitter (KdS) 4-metric reads, 
\begin{equation}
\label{KdS4}
\wt ds_4^2 =  \wt g_{\mu\nu}^K dx^{\mu}x^{\nu} =  - \frac{\sigma}{\rho^2} \left(dt - \frac{a \sin^2 \theta}{\kappa} d\varphi \right)^2 
+ \frac{\rho^2}{\sigma} dr^2 + \frac{\rho^2}{\chi}
d\theta^2 + \frac{\chi \sin^2 \theta}{\rho^2}  
 \left(a dt - \frac{r^2 + a^2}{\kappa} d\varphi \right)^2 
\end{equation}
in terms of ``Boyer-Lindquist''-coordinates with constants $m$, $a$, 
and $\kappa  =  1 + \Lambda a^2/3$, 
and in terms of the functions
\begin{equation}
\label{func}
 \sigma  =  (r^2 + a^2)(1 - \frac{\Lambda r^2}{3}) - 2mr,  \qquad \rho^2 = r^2 +
 a^2 \cos^2 \theta,  \qquad   \chi = 1 + \frac{\Lambda a^2 \cos^2 \theta}{3}
\end{equation}
where $\sigma$ generalises the synonymous function in 3.3.1.

The constants $m$ and $a$  satisfy bounds given by the extreme solutions, but also ``absolute'' 
bounds in terms of $\Lambda$ alone, cf. \cite{AM,CCK,CKN,GRS} for a discussion.

The twist scalar of (\ref{KdS4}) as defined after Equ (\ref{twi}) takes the form
\begin{equation}
\widetilde \omega_K = - 2 J_K \left(\cos^3 \theta - 3 \cos\theta - \frac{a^2 \cos\theta \sin^4 \theta}{\rho^2} \right) 
\end{equation}
where
\begin{equation}
\label{JKdS}
J_K = \frac{ma}{\kappa^2}
\end{equation}

is the angular momentum as defined in Sect. \ref{ang}.
The bounds on $m$ and $a$ entail a bound on $J_K$, again in terms of the
angular momenta $J_E$ of the 1-parameter family of the extreme solutions;
the latter satisfy the absolute bound $\Lambda J_E \le \Lambda J_{max} \approx 0.170$
saturated for one particular extreme solution.
Interestingly, this value exceeds the  non-existence bound of 
Thm. 3.2. (valid for TPSM solutions of (\ref{LichBY}) on the round 
hypertorus (\ref{Nar})).  

We now recall how (\ref{KdS4}) arises from $(t,\varphi)$-symmetric data.  
The slice $t=const$ is maximal with induced metric

\begin{equation}
\label{KdS3}
ds^2_K = g_{ij}^K dx^i dx^j = \rho^2 \left(\frac{dr^2}{\sigma} +
\frac{d\theta^2}{\chi} \right) +  
\frac{\sin^2 \theta}{\kappa^2 \rho^2} \left[\chi \left(r^2 + a^2 \right)^2 -
\sigma a^2 \sin^2 \theta \right] d\varphi^2.  
\end{equation}

 This metric still fits on a 3-manifold ${\cal M}$ of topology $\mathbb{S}^2 \times \mathbb{S}$ \cite{JC} 
which can be seen as in the Kottler case.
As before we restrict ourselves to the region $r \in [r_b, r_c]$ where $\sigma \ge 0$,
which is bounded by the black hole- and the cosmological horizon.
At $\sigma = 0$ the metric (\ref{KdS3}) is regularised by replacing $r$ by $\alpha$ 
defined via (\ref{alp}) but with $\sigma$ from (\ref{func}).  
 Thus the metric becomes periodic in $\alpha$ with period
\begin{equation}
P(\Lambda, m,a) := 2 \int_{r_b}^{r_c} \sigma^{-\frac{1}{2}} dr. 
\end{equation} 
As we are not interested in the complete set of solutions here, we take $P$ equal to the circumference 
$T$ of ${\cal M}$ which leaves us with just one pair of horizons. 

From Proposition \ref{tp2}, and using Remark 2. after this Proposition, we can now determine
the second fundamental form $ K_{ij}^K$ via $\omega_K = \wt \omega_K |_{\cal M}$  

With these preparations we now construct new data as follows.
\begin{defn}
\label{KKdS}
We set out from non-extreme KdS data $({\cal M}, g_{ij}^K, K_{ij}^K)$ as SM
with angular momentum $J_K$ and with $K_{ij}^K$ generated via (\ref{Kom}) from $\omega_K$. 
For some $J \in \mathbb{R}$ we define $\omega(J) = J \omega_K/J_K$ and
$K_{ij}(J)$ again via (\ref{Kom}).  
\end{defn}

It follows that $K_{ij}(J) =  K_{ij}^K J /J_K$ and 
$\Omega^2(J) = K_{ij}(J) K^{ij}(J) = J^2 \Omega_K^2/J_K^2$. 
Note that (\ref{JKdS}) does {\it not} hold for $J \neq J_K$.  

Applying Premoselli's theorem (Thm. \ref{bp}) and Propositions 1, 3 and 4 now yields 

\begin{thm}
\label{os}
In the setting described in Definition \ref{KKdS}.  we claim
\begin{enumerate}
\item  
There exists $J_* \ge J_K$ such that, for $J \in (0, J_*)$,  
Equ. (\ref{phi}) with $\Omega(J)$ has at least two positive
solutions, one of which is minimal and stable. Moreover, for $J = J_*$ there is a unique 
marginally stable solution, for $J$ slightly below $J_*$ there are precisely two
solutions, and for $J > J_*$ there is no solution.
In the limit $J \rightarrow 0$, the family of minimal, 
stable solutions tends to zero like $(\Lambda J)^{1/4}$.
\item
If 
\begin{equation}
\label{Omb}
\int_{\cal M} \Omega_K^2 dv < \Lambda .V_K 
\end{equation}
 for KdS data with $J_K$, $\Omega_K$ and volume $V_K$,  
it follows that $J_* > J_K$.
\item Inequality (\ref{Omb}) (and therefore the conclusion of point 2.) 
always holds for sufficiently small $J_K$.       
\end{enumerate}
\end{thm}
 
\begin{proof}
We first note that $K_{ij}(J)$ is smooth for all $J \in \mathbb{R}$.
Point 1. follows now trivially from the results stated before the
theorem and makes no direct reference to the properties of KdS.
We have stated this point explicitly to illustrate how 
 Thm. \ref{bp} allows to deduce the existence of a large family of
 solutions from  a single one. 
Regarding the case $J = 0$ we note that, apart from $\phi \equiv
 0$,  non-trivial solutions definitely exist as well, namely  solutions 
to the Yamabe problem \cite{LP}; 
however, we have no information about their multiplicity and propagation 
to $J \neq 0$ here.   
 To prove 2. we conclude indirectly: Assume that $J_*= J_K$ which means
 that, within the 1-parameter family of $K_{ij}(J)$-tensors generated 
from $\omega(J)$ with $J \in \mathbb{R}$ as in Definition \ref{KKdS}, 
the KdS tensor $K_{ij}^K$ had in fact the maximal angular momentum permitted by
Premoselli's theorem. Then 2. of that theorem would imply that the KdS data are marginally
 stable. However, (\ref{Omb}) together with Proposition \ref{stest} implies strict instability, 
a contradiction. (Note that this Proposition applies here directly since KdS are ID rather than just a SM). 

For the final point 3. it suffices to show that $\Omega_K^2 < \Lambda$ for
small $J_K$. This is intuitively clear as $\Omega_K^2$ is of order $J_K^2$ near 
$J_K = 0$. To see this in detail, it is useful to rescale all variables and constants to the dimensionless quantities 
\begin{equation}
\label{resc}
\bar r = \sqrt{\Lambda} r \qquad  \bar \theta = \theta \qquad \bar \varphi =
\varphi, \qquad \bar a = \sqrt{\Lambda} a  \qquad \bar m = \sqrt{\Lambda} m  \qquad \bar J_K = \Lambda J_K   
\end{equation}

In terms of these variables, the KdS metric and the terms characterising rotation take the forms 
\begin{equation}
ds^2_K  =  \Lambda^{-1} \bar ds^2_K \qquad \omega_K  =  \Lambda^{-1} \bar \omega_K
\qquad \Omega_K^2  =  \Lambda \bar \Omega_K^2
\end{equation}
where all quantities on the l.h.s. are functions of $\Lambda, m, a, r$ and $\theta$, while all
 quantities with bars can be written in terms
of the functions $\bar r, \bar \theta$ and constants $\bar m$ and $\bar a$ only (and hence do not depend explicitly on $\Lambda$). 

Using (\ref{Om}) we now observe that $\bar \Omega_K$ can be written as 
$\bar \Omega_K(\bar r, \bar \theta, \bar m, \bar a) =  \bar J_K \bar \Omega_{K,0}(\bar r, \bar \theta, \bar m, \bar a)$. 
We first show that $\bar \Omega_{K,0}^2$ is analytic in all arguments. 
This follows from the fact the denominator in (\ref{Om}) only vanishes on the axis where, 
however, it is ``regularized'' by the zeros of the numerator in the same manner as in the Bowen-York example,
and we still have $\bar \Omega_K^2 = f \sin^2 \bar \theta$ near the axis with
an analytic function $f$.
Next we recall that, for regular KdS data, $\bar m$ and $|\bar a|$ are bounded
from above (by a certain number). Therefore, $\bar \Omega_{K,0}^2(\bar r, \bar\theta, \bar m, \bar a)$,
being an analytic function on a compact domain, is bounded from above (by a number). 
 This implies that $\bar \Omega_K=  \bar J_K \bar \Omega_{K,0}$ can be made as small as 
needed by decreasing $\bar J_K$. 

We conclude that  the KdS metric can indeed be ``overspun'' in the sense
that, for small $J_K$, one can put more angular momentum than (\ref{JKdS}) on
the background (\ref{KdS3}). (We remark that our notion of ``overspinning'' has
nothing to do with attempts of exceeding the angular momentum limit 
of extreme Kerr black holes, cf. \cite{CB}). 
\end{proof}

We finish with two remarks.

\begin{description}

\item[ ``Conformally relaxed'' Kerr-de Sitter data.]
Premoselli's theorem and the previous one imply that for any KdS seed data with 
$J_K$ small enough, there exist stable, minimal data 
which are conformal to these KdS data. We call them ``conformally relaxed'',
(while the unstable KdS data themselves are ``conformally excited'').  
 In view of their stability, they will necessarily (by Proposition
 \ref{csym}) inherit the axisymmetry of the KdS seed, 
and they will have the same angular momentum by virtue of the conformal
 invariance of (\ref{am}).
These data will very likely be non-stationary - in any case, again due to
their stability, they cannot be a member of the KdS family with different
parameters, as follows from the proof of the above theorem. 
In case the data are really dynamic, it would be interesting to determine their
time evolution. Due to axisymmetry, this evolution would preserve 
the angular momentum. Therefore, it is not impossible that the resulting
spacetime settles down to the same member of the KdS family one started with.  

\item [Data with marginally outer trapped surfaces.]
Within the conformal method, there has been considered the problem of finding
ID which contain (marginally) outer trapped surfaces ((M)OTS). This was 
accomplished by imposing suitable boundary conditions on the SM, first in the  asymptotically flat context 
\cite{DAIN,DM}, but recently also for compact manifolds with boundary
\cite{HM,HT}. However, this work does not cover the Lichnerowicz equation 
with the present sign of the coefficient of $\phi^5$.  It would be
interesting to handle this case as well \cite{GT}. This seems to require 
combining the techniques of the aforementioned papers with those of Hebey et al. \cite{HPP}
 and/or  Premoselli \cite{BP}.
 
In the examples considered in this paper, both the Kottler as well as the
KdS family of data contain minimal surfaces. In the former case, the
minimal surfaces are expected to turn into MOTS when a 
small angular momentum is added. In the same manner, MOTS should arise
when the angular momentum of KdS is slightly reduced or increased beyond the
stationary value as described above. It would be interesting to settle this 
in the generic case.        
\end{description}
~\\
 {\large\bf Acknowledgements.}
We are grateful to Patryk Mach for verifying numerically point 3 in 
Conjecture 1, and to Lars Andersson, Robert Beig,  Piotr Chru\'sciel, Emmanuel
Hebey, Patryk Mach  and Bruno Premoselli for helpful discussions and correspondence.  
The research of P.B., and a visit of W.S. to Cracow, were supported  by the Polish National Science 
Centre grant no. DEC-2012/06/A/ST2/00397.
The research of W.S. was funded by the Austrian Science Fund (FWF): P23337-N16

\end{document}